\documentclass[10pt, nofootinbib,longbibliography, aps, prb, reprint, notitlepage,superscriptaddress]{revtex4-2}

\usepackage[latin1]{inputenc}
\usepackage{dsfont}
\usepackage{graphicx}
\usepackage{amsmath}
\usepackage{amsthm}

\usepackage{amsfonts}
\usepackage{amssymb}
\usepackage{mathtools}
\usepackage[colorlinks=true,citecolor=blue]{hyperref}
\usepackage[shortlabels]{enumitem}
\usepackage[capitalise,compress]{cleveref}
\usepackage{textcomp, complexity}
\usepackage{silence}

\newtheorem{theorem}{Theorem}
\newtheorem{problem}{Problem}
\newtheorem{lemma}{Lemma}
\newtheorem{definition}{Definition}
\newtheorem{corollary}{Corollary}
\newcommand{\avg}[1]{\left \langle #1 \right\rangle}
\newcommand{\ket}[1]{\left | #1 \right\rangle}
\newcommand{\bra}[1]{\left \langle #1 \right |}

\newcommand{\braket}[2]{\left\langle #1|#2\right\rangle}
\newcommand{\abs}[1]{\left | #1 \right|}
\renewcommand{\epsilon}{\varepsilon}
\renewcommand{\O}[1]{\mathcal O\left(#1\right)}
\newcommand{\id}{\mathds{1}}
\newcommand{\norm}[1]{\left\|#1\right\|}



\usepackage{braket}

\usepackage{array}

\usepackage[dvipsnames]{xcolor}

\usepackage{mathtools}
\DeclarePairedDelimiter\ceil{\lceil}{\rceil}
\DeclarePairedDelimiter\floor{\lfloor}{\rfloor}

\usepackage{ulem}
\usepackage{times}

\begin{document}
\title{Simulation Complexity of Many-Body Localized Systems}
\author{Adam Ehrenberg}
\affiliation{Joint Quantum Institute, NIST/University of Maryland, College Park, MD 20742, USA}
\affiliation{Joint Center for Quantum Information and Computer Science, NIST/University of Maryland, College Park, MD 20742, USA}
\author{Abhinav Deshpande}
\affiliation{Institute for Quantum Information and Matter, California Institute of Technology, Pasadena, CA 91125, USA}
\author{Christopher L.\ Baldwin}
\affiliation{Joint Quantum Institute, NIST/University of Maryland, College Park, MD 20742, USA}
\affiliation{National Institute of Standards and Technology, Gaithersburg, MD 20899, USA}
\author{Dmitry A.\ Abanin}
\affiliation{Department of Theoretical Physics, University of Geneva, 1211 Geneva, Switzerland}
\author{Alexey V.\ Gorshkov}
\affiliation{Joint Quantum Institute, NIST/University of Maryland, College Park, MD 20742, USA}
\affiliation{Joint Center for Quantum Information and Computer Science, NIST/University of Maryland, College Park, MD 20742, USA}
\date{\today}
\begin{abstract}
We use complexity theory to rigorously investigate the difficulty of classically simulating evolution under many-body localized (MBL) Hamiltonians.
Using the defining feature that MBL systems have a complete set of quasilocal integrals of motion (LIOMs), we demonstrate a transition in the classical complexity of simulating such systems as a function of evolution time.
On one side, we construct a quasipolynomial-time tensor-network-inspired algorithm for strong simulation of 1D MBL systems (i.e., calculating the expectation value of arbitrary products of local observables) evolved for any time polynomial in the system size. 
On the other side, we prove that even weak simulation, i.e. sampling, becomes formally hard after an exponentially long evolution time, assuming widely believed conjectures in complexity theory. 
Finally, using the consequences of our classical simulation results, we also show that the quantum circuit complexity for MBL systems is sublinear in evolution time.
This result is a counterpart to a recent proof that the complexity of random quantum circuits grows linearly in time.
\end{abstract}

\vspace*{0in}
\maketitle
\section{Introduction}\label{sec:Introduction}
As quantum computers become larger-depth, less error-prone, and eventually fully fault-tolerant, it will become increasingly important to understand which computational problems admit quantum speedups over the best possible classical algorithms.
This question broadly falls under the domain of computational complexity theory, which studies how easy or hard it is to solve certain problems under various computational assumptions. 
More specifically, \emph{sampling complexity}, the study of how difficult it is to draw samples from classes of probability distributions, is a useful framework for studying the classical hardness of simulating quantum systems, and can help to narrow the parameter space where quantum advantage may be obtained.
At their core, many quantum experiments reduce to repeatedly preparing a certain quantum state, measuring it (thus generating a probability distribution of outcomes), and classically post-processing on the measurement results.
This high-level viewpoint motivates the systematic study of quantum systems via the lens of sampling complexity.
Indeed, the past ten years have seen significant interest in sampling after the proof (up to widely believed mathematical conjectures) that one could obtain a quantum advantage in the famous Boson Sampling problem~\cite{Aaronson2011}, leading to the recent demonstration of quantum sampling experiments believed to be beyond the accessibility of classical simulations ~\cite{arute_quantum_2019,zhong_quantum_2020, zhong_phase-programmable_2021}. 

With the same motivation in mind, Ref.~\cite{Deshpande2018Dynamical} considered a system of indistinguishable non-interacting bosons distributed on a lattice and evolved under a local Hamiltonian (also see Refs.~\cite{Muraleedharan2018,Maskara2019a} for variants of this problem).
Intuitively, one expects that classical simulation is initially easy while the particles are separated, but grows more difficult as the system evolves.
Reference~\cite{Deshpande2018Dynamical} formalized this idea by showing that sampling remains easy until the particles have evolved for long enough to travel the distance initially separating them, whereafter their fundamental indistinguishability leads to quantum interference that is hard to classically simulate. 
A key corollary of this result is that classical sampling is easy in single-particle-localized systems, where the particle wavepackets do not spread out~\cite{Thouless1974Electrons,Kramer1993Localization,Billy2008Direct,Roati2008Anderson}. 
Thus, while single-particle localized systems are fascinating from a condensed matter perspective, we do not necessarily expect them to encode hard computational problems, and we will likely have to look to other types of systems to find useful quantum speedups.

The present work is concerned with the more subtle situation of \emph{many-body} localization (MBL)~\cite{Nandkishore2015Many,Abanin2017Recent,Abanin2019} in spin systems, which we take to mean any spin Hamiltonian having a complete set of local integrals of motion (precisely defined below)~\cite{Serbyn2013Local,huse_2014_phenomenology,Chandran2015Constructing,Ros2015Integrals,imbrie_many-body_2016}. These systems differ from the single-particle-localized situation described above in a crucial way: the quasilocal commuting operators that fully describe the dynamics of these systems interact with one another through nontrivial exponentially decaying interactions.
These interactions can spread entanglement through the system and destroy separability of an initial state over exponentially long time-scales. 

Suppose we time-evolve an initial product state under an MBL Hamiltonian acting on $N$ spins and then measure the result in a product basis, generating a probability distribution.
We will explore the algorithmic time complexity of both \emph{strong simulation} and \emph{weak simulation} of this physical system. 
Weak simulation is the ability to sample from the distribution of outcomes, whereas strong simulation is the ability to calculate all marginal and conditional probabilities of the outcomes. 
The ability to strongly simulate a system implies the ability to sample from it \cite{Terhal2002a}, but not vice versa---one can, in principle, sample from a distribution without ever knowing the values of the probabilities.

Observe that in describing the problem of interest, we have introduced two types of time: evolution and computational. For clarity in the remainder of this work, we will use a lower-case $t$ to refer to the physical evolution time, or the time for which the MBL Hamiltonian acts on the initial state. We denote the time complexity of a classical algorithm for a given simulation task with an upper-case $T$. 

We now present our main results. Using techniques inspired by tensor networks, we present an algorithm that can strongly simulate (and thus sample from) any one-dimensional MBL system in quasipolynomial computer time (i.e., times of the form $T = \exp{[\O{\log^{c}N}]}$ \cite{asymptotic_notation} for some $c > 1$), for any evolution time $t$ polynomial in the system size $N$.
It is interesting that even this algorithm does not run in strictly polynomial time, and we are not aware of any algorithm which (provably) can.
Conversely, by using ideas inspired by the hardness of the Instantaneous Quantum Polynomial (IQP) sampling problem in Ref.~\cite{Bermejo-Vega2018}, we also show that the MBL sampling problem becomes hard in the worst case after evolution time $t = \Omega(\exp{[{N^{\delta}}]})$ for arbitrarily small $\delta > 0$ (by ``worst case,'' we mean that we demonstrate that a specific family of MBL Hamiltonians becomes hard to simulate, but this family does not contain all possible MBL Hamiltonians). 
These results are summarized  in \cref{tab:results_summary}.

\begin{table}
    \centering
    \begin{tabular}{c|c|c}
         Evolution Time $t$ & Complexity & Task \\ \hline
         $\O{\log N}$ & Easy \cite{Osborne2006} & Strong Simulation \\
         $\O{\mathrm{poly} N}$ & Quasi-easy & Strong Simulation \\
         $\O{\mathrm{quasipoly} N}$ & Quasi-easy & Strong Simulation \\
         $\Omega({\exp N})$ & Hard & Weak Simulation
    \end{tabular}
    \caption{Summary of our results for classical simulation. We define ``quasi-easy'' to be those problems admitting a quasipolynomial-time algorithm but which may yet possess a polynomial-time algorithm.}
    \label{tab:results_summary}
\end{table}

Interestingly, as a consequence of our proof techniques, we can also derive results on the \emph{quantum circuit complexity} of implementing time evolution due to an MBL Hamiltonian.
The quantum circuit complexity of a unitary $U$ is the minimum number of gates (from a predefined universal gate set) required to approximate $U$.
In many-body physics, it is of great significance to understand how the quantum circuit complexity of a time-evolution operator $e^{-iHt}$ grows with respect to the time $t$ for various Hamiltonians $H$.
In the context of high-energy physics, gravitational physics, and the AdS/CFT correspondence, it was conjectured \cite{Brown2016,Brown2016a} that the circuit complexity of a conformal field theory is dual to the action of a gravitational theory describing the bulk.
More specifically, it has been conjectured that the circuit complexity of fast-scrambling dynamics grows linearly in time until a timescale exponential in system size.
This conjecture has gathered support due to recent work \cite{Brandao2021,Haferkamp2021a}.
In stark contrast with these fast scramblers, we show in this work that the circuit complexity for sufficiently localized MBL Hamiltonians grows only sublinearly with evolution time.
Therefore, our work suggests that, in addition to classical complexity, studying the quantum complexity of simulating time evolution can also serve as a basis for classifying the ergodicity of quantum dynamics.

Others have investigated the simulation of MBL systems. For a few examples, see Refs.~\cite{weidingerSelfconsistentHartreeFockApproach2018, detomasiEfficientlySolvingDynamics2019,chandranSpectralTensorNetworks2015d, pollmannEfficientVariationalDiagonalization2016, wahlEfficientRepresentationFully2017}, which introduce efficient methods for classically simulating both spin and weakly-interacting fermionc MBL systems. However, while these works demonstrate empirically good numerical alternatives to computationally demanding exact diagonalization schemes, they stop short of formal proofs that these algorithms can maintain accuracy for all MBL systems as the system size grows (though Ref.~\cite{chandranSpectralTensorNetworks2015d} does contain some formal proofs in the case of exactly local integrals of motion, as opposed to the more general quasilocal integrals of motion we consider here). Overall, our work is the first to systematically investigate the simulation of generic MBL systems from a rigorous complexity-theoretic perspective. 

The rest of the paper is organized as follows.
In \cref{sec:Setup}, we formally define the simulation problem. 
We then prove in \cref{sec:Truncation} crucial mathematical results that we use in \cref{sec:Easiness} to demonstrate the quasipolynomial runtime of our tensor-network algorithm for strong simulation.
Correspondingly, in \cref{sec:Hardness} we demonstrate that generic MBL Hamiltonians are hard to sample from after exponentially long evolution time $t$.
In \cref{sec:Quantum} we also show that that the quantum circuit complexity of the time-evolution operator of a sufficiently localized MBL Hamiltonians is sublinear in time.
Finally, in \cref{sec:Conclusion} we synthesize these results and consider directions for future work. 


\section{Setup}\label{sec:Setup}
Consider a 1D lattice of $N$ spin-1/2 particles (with spin operators $\sigma_{i}^{\alpha}$, $\alpha = x, y, z$) that evolve under some Hamiltonian $H$.
We say that $H$ is MBL if there exists a quasilocal unitary $U$ (defined below) that brings $H$ to the form
\begin{equation}
	H = \sum_{i}J_{i}\tau_{i}^{z} + \sum_{i < j}J_{ij}\tau_{i}^{z}\tau_{j}^{z} + \sum_{i < j < k}J_{ijk}\tau_{i}^{z}\tau_{j}^{z}\tau_{k}^{z} + \dots, \label{eqn:HMBL}
\end{equation}
with $[\tau_{i}^{z}, \tau_{j}^{z}] = 0$ and $\abs{J_{i_{1}\dots i_{p}}} \leq\exp\left(- {(i_{p} - i_{1})}/{\xi}\right)$. We call the $\sigma_{i}^{z}$ the \emph{physical bits} (p-bits) because they represent the experimentally accessible basis of observables, and we call the $\tau_{i}^{z}$ the \emph{local integrals of motion} (LIOMs) or \emph{localized bits} (l-bits) because they commute with the Hamiltonian and thus represent a set of $N$ conserved quantities that constrain the dynamics.

We define a quasilocal unitary, which we schematically depict in \cref{fig:nlql}, as follows:
\begin{definition}[Quasilocal unitary \cite{Abanin2019}]\label{def:QLU}
     A unitary $U$ is quasilocal if it can be decomposed on a finite 1D lattice with $N$ sites as
	\begin{equation}
 	U = \prod_{n=1}^{N}\prod_{j=1}^{n}\prod_{i=0}^{\lfloor (N-n)/n \rfloor}U_{in + j}^{(n)},
 	\end{equation}
 	where $U_{k}^{(n)}$ acts on sites $k, k+1, \dots, k+n-1$ such that
	\begin{equation}\label{eqn:frob_closeness}
	\norm{\id - U_{k}^{(n)}}^{2} < q e^{-\frac{(n-1)}{\xi}},
	\end{equation}
	where $\norm{\cdot}$ is the operator norm (\textit{i.e.}, the largest singular value of the operand) and $q$ is some $\O{1}$ constant \cite{families}. When $k + n - 1 > N$, $U_{k}^{(n)}$ should be interpreted as a tensor product of two unitaries, one acting on sites $k$ through $N$, and the other on $1$ to $k + n - 1 - N$. 
\end{definition}
This means that we can decompose $U$ into a sequence of $n$ layers of $n$-site unitaries, where the more sites a constituent unitary acts on, the closer it is to the identity.
We call $U$ ``quasilocal'' because, though any two distant sites may be entangled, the amount of entanglement generated decays rapidly with distance.
\begin{figure}
    \centering
    \includegraphics[width=\linewidth]{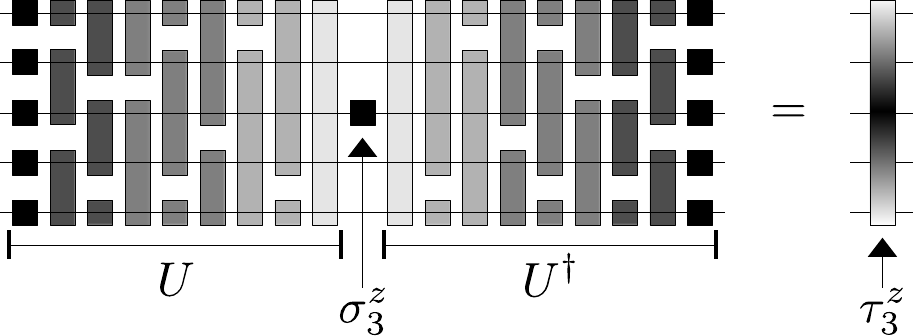}
    \caption{Schematic depiction of a quasilocal unitary $U$ on $N = 5$ sites converting between the physical and localized bases, $U\sigma_{3}^{z}U^{\dag} = \tau_{3}^{z}$. As described in \cref{def:QLU}, $U$ decomposes into constituents, and the opacity of each constituent block represents its proximity to the identity with respect to the norm $\norm{\cdot}$; the lighter the block, the closer it is to the identity.}
    \label{fig:nlql}
\end{figure}

Having defined the properties of our Hamiltonian $H$, consider now an experiment whereby the system is initially prepared in the physical state $\ket{0\dots0}$ (i.e., $\forall i$ $\sigma_{i}^{z}\ket{0\dots0} = \ket{0\dots0}$), then time-evolved into $e^{-iHt}\ket{0\dots0}$, and finally measured in the physical basis. 
The probability of observing an outcome $\ket{\sigma}$ after a time $t$ is $\mathcal{D}(\sigma) \equiv \abs{\braket{\sigma|e^{-iHt}|0\dots0}}^{2}$.
As previously discussed, we want to assess the difficulty of both drawing a sample from (weak simulation) and calculating marginals of (strong simulation) the distribution $\mathcal{D} \equiv \{\mathcal{D}(\sigma)\}_{\sigma}$. 
However, even a quantum computer directly performing such an experiment will be subject to at least small errors, and will thus be unable to draw a sample from this distribution perfectly.
Therefore, we will only assess the difficulty of \emph{approximate} sampling from a distribution $\mathcal{D}_{\epsilon}$ that is $\epsilon$-close to $\mathcal{D}$ in total variation distance (TVD):
\begin{equation}
    \norm{\mathcal{D}_{\epsilon}-\mathcal{D}}_\mathrm{TVD} = \frac{1}{2}\sum_{\sigma}\abs{\mathcal{D}_{\epsilon}(\sigma)-\mathcal{D}(\sigma)} < \epsilon. 
\end{equation}
We state our sampling problem formally.
\begin{problem} \label{prob:mblsampling}
Let $H$ be an MBL Hamiltonian (according to the above definition) on an $N$-site chain and $U$ its corresponding quasilocal unitary.
Consider the distribution $\mathcal{D} ~=~ \{\abs{\braket{\sigma|e^{-iHt}|0\dots0}}^{2}\}_{\sigma}$.
Given a description of $H$ in terms of physical operators, an efficient algorithm to compute any element of any constituent $U_{k}^{(n)}$ of $U$, and an efficient algorithm to compute any coupling $J_{i_{1}\dots i_{p}}$,
output a sample from a distribution $\mathcal{D}_{\epsilon}$ that is $\epsilon$-close to $\mathcal{D}$ in total variation distance for any $\epsilon > 0$. 
\end{problem}
A few comments on \cref{prob:mblsampling} are worthwhile. 
We need these efficient algorithms to calculate any desired constituent $U_{k}^{(n)}$ and any desired coupling $J_{i_{1}\dots i_{p}}$ because knowledge of these quantities will be crucial for our algorithm, and it is too computationally expensive to calculate and naively list out all exponentially many of them. Formally, we assume that we have an \emph{oracle} for these properties of the system. 
  
Ideally we would be able to extract $J_{i_{1}\dots i_{p}}$ and $U$ efficiently from the description of $H$ in the physical basis. However, MBL is typically considered in the context of disordered spin chains where it may not always be possible to efficiently compute these quantities (though there is some evidence that this may be possible -- see Refs. \cite{chandranSpectralTensorNetworks2015d, pollmannEfficientVariationalDiagonalization2016, wahlEfficientRepresentationFully2017, kulshreshthaApproximatingObservablesEigenstates2019,Chertkov2021}). Therefore, we do not restrict ourselves to this particular mechanism for producing LIOMs, and our results will apply to any Hamiltonian that can be diagonalized by quasilocal unitary $U$ into the form \cref{eqn:HMBL}.
Finally, neither the specific initial state nor the measurement basis are critical to our formulation of \cref{prob:mblsampling} as long as they are a product state and a product basis.
This is because we allow $U$ to contain a layer of $\O{1}$ 1-site terms so that we do not pick out any particular basis as special. 
Our main results concern the classical time complexity $T$ of solving \cref{prob:mblsampling} as a function of evolution time $t$ and system size $N$.


\section{Truncating the Canonical Hamiltonian}\label{sec:Truncation}
We proceed to characterize the classical complexity of solving \cref{prob:mblsampling} in two ways depending on the evolution time $t$. 
If $t = \O{\log N}$ and $H$ is finite-range in the physical basis, Ref.~\cite{Osborne2006} proves there exists an efficient matrix-product operator representation of the propagator $e^{-i H t}$. 
This representation may be used to approximately sample from the outcome distribution of evolution under $H$. 
See the Supplemental Material \cite{SM} for more details.

For longer times $t = \omega(\log N)$, we construct a Hamiltonian $\tilde{H}$ for which the time-evolved probability distribution is $\mathcal{\tilde{D}}\equiv \{|\braket{\sigma|e^{-i\tilde{H}t}|0\dots0}|^{2}\}$, such that (a) $\|\mathcal{D}-\mathcal{\tilde{D}}\|_\mathrm{TVD} \leq \epsilon$ and (b) the distribution associated with evolution under $\tilde{H}$ can be sampled from in computer time scaling quasipolynomially with the number of spins $N$.
The total variation distance between the probability distributions associated with two pure states $\ket{\psi}$ and $\ket{\phi}$ can be upper bounded by the 2-norm distance~\cite{Arkhipov2015}, which in turn can be bounded~\cite{Maskara2019a} as $\norm{\ket{\psi(t)}-\ket{\phi(t)}}_{2} \leq ||H-\tilde{H}||t \equiv\norm{\Delta H}t$, where $\norm{\cdot}$ is the standard operator norm. Therefore, if we want the two distributions to be $\epsilon$-close in total variation distance up to a time $t$, it is sufficient to ensure $\norm{\Delta H} \leq {\epsilon}/{t}$. 

We construct this approximate Hamiltonian $\tilde{H}$ by \emph{truncating} the exact Hamiltonian in two ways: via the coupling constants and the LIOMs.
In particular, we set the coupling constants equal to zero if they connect sites beyond a certain radius $r_{J}$, and we set equal to the identity those constituents of $U$ supported on more than $r_{U}$ sites.
Mathematically:
\begin{equation}\label{eq:trunc_J}
	\tilde{J}_{i_{1}\dots i_{p}} = 
	\begin{cases}
		J_{i_{1}\dots i_{p}} & \text{if } i_{p}-i_{1} < r_{J} \\
		0 & \text{if } i_{p}-i_{1} \geq r_{J}
	\end{cases},
\end{equation}
\begin{align}
	\tilde{U} &= \prod_{n=1}^{r_{U}}\prod_{j=1}^{n}\prod_{i=0}^{\lfloor (N-n)/n \rfloor}U_{in + j}^{(n)}, \label{eq:trunc_U}\\
	\tilde{\tau_{i}}^{z} &= \tilde{U}\sigma_{i}^{z}\tilde{U}^{\dag}. 
\end{align}
We can now bound the norm of 
\begin{equation}
    \Delta H \equiv H-\tilde{H} = \sum_{I} J_{I}\tau_{I}^{z} - \tilde{J}_{I}\tilde{\tau}_{I}^{z}
\end{equation}
by applying the triangle inequality:
\begin{align}
	\norm{\Delta H} &\leq 
	\sum_{I} \Big( |J_{I}-\tilde{J}_{I}| + |\tilde{J}_{I}|\norm{\tau_{I}^{z} - \tilde{\tau}_{I}^{z}} \Big) \label{eqn:hamiltonian_difference_norm_split},
\end{align}
where we have introduced $I$ as a general multi-index for brevity.
Before continuing, it is useful to define $S_{p, n_{0}} \equiv \sum_{n = n_{0}}^{\infty}\binom{n}{p}e^{-\frac{n}{\xi}}$.
Intuitively, this sum appears because we will often be interested in summing over couplings of a range exceeding some $n_{0}$, and each coupling comes with an associated exponential decay. 
Assuming that the localization length $\xi < 1/\log 2$, we have:
	\begin{equation}\label{eq:spn0}
		S_{p,n_{0}} \leq C
		\begin{cases}
			e^{-\frac{n_{0}}{\xi}} & p = 0 \\
			pe^{-ap} & n_{0} < n_{*}, p > 0 \\
			\frac{n_{0}^{p+1}\sqrt{p}}{p!}e^{-\frac{n_{0}}{\xi}} & n_{0} \geq n_{*}, p>0
		\end{cases},
	\end{equation}
where $a \equiv \log (e^{1/\xi}-1)$, $n_{*} \equiv p(1-e^{-1/\xi})^{-1}$, and $C$ is some $\O{1}$ constant.
See \cref{lem:sum_bound_full} in the Supplemental Material~\cite{SM} for a detailed proof.

We now separately bound the two contributions to \cref{eqn:hamiltonian_difference_norm_split}.
The details, which are in the Supplemental Material~\cite{SM}, make heavy use of \cref{eq:spn0}, and the result is 
\begin{equation}
\norm{\Delta H} \leq C_{J}Nr_{J}e^{-kr_{J}} + C_{U}N^{2}e^{-\frac{r_{U}}{2\xi}}, \label{eqn:error}
\end{equation}
where $C_{U}$, $C_{J}$, and $k$ are constants independent of $N$.
Intuitively, the factors of $N$ come from summing over sites, and the exponential decay factors come from the decay properties of $H$ and $U$. 
To ensure that $\norm{\Delta H} \leq \epsilon/t$ for some polynomially long time $t = \mathcal{O}(N^{b})$, it suffices to choose 
\begin{equation}\label{eq:rUrJ}
	r_{U} = \Omega (\xi b \log N),
	r_{J} = \Omega (bk^{-1}\log N).
\end{equation}
Therefore, truncating the coupling coefficients and the diagonalizing quasilocal unitary to a scale logarithmic in the system size is sufficient to produce a distribution that is close in total variation distance to the true distribution.


\section{Quasipolynomial-Time Sampling}\label{sec:Easiness}
Having defined an appropriate approximation $\tilde{H}$ we now describe how to sample from the distribution generated by $\tilde{H}$. 
More precisely, we provide an algorithm for strong simulation, meaning it can calculate all probabilities and marginal probabilities of the distribution generated by measuring the simulated system in any local basis. Equivalently, it can estimate the expectation value of arbitrary products of local observables.  
Strong simulation implies the ability to solve the easier problem of weak simulation, i.e.~sampling, which itself implies the ability to calculate the expectation values of local observables \cite{Terhal2002a}. 
Specifically, our algorithm will calculate
\begin{equation}\label{eqn:observable_sampling}
   \langle\tilde{O}\rangle_{t} = \bra{\psi(0)}e^{i\tilde{H}t}Oe^{-i \tilde{H}t}\ket{\psi(0)},
\end{equation}
where $O$ is a product of single-site observables in the p-bit basis of the form $O = \sigma_{i}^{z}\prod_{j<i}P_{j}$, with $P_{j}$ a projector of qubit $j$ onto the 0 or 1 outcome when measuring in the appropriate local basis, and the tilde indicates that we evolve with the approximate Hamiltonian $\tilde H$. 
Intuitively, $O$ is selected such that \cref{eqn:observable_sampling} calculates the conditional probability $P(z_{i}|z_{i-1}\dots z_{1})$, and drawing a sample given these conditional probabilities is equivalent to flipping $N$ biased coins, where the bias of each coin is conditioned on the previous outcomes. 
For $t = \O{\log N}$, we use the algorithm implied by results in Ref.~\cite{Osborne2006} and  elucidated in the Supplemental Material~\cite{SM}.
In short, when $H$ is short-range in the p-bit basis, the propagator for the true Hamiltonian $e^{-iHt}$ can be efficiently approximated by a matrix product operator $M$.
Because a product of local observables also admits a matrix product operator form, $\bra{\psi(0)|M^{\dag}OM\ket{\psi(0)}}$ may be calculated in computational time $T = \mathcal{O}(\mathrm{poly} N)$.

For the more complicated problem of $t = \omega(\log N)$, we provide a different algorithm where each unitary in the circuit is interpreted as a tensor, making the quantum circuit for time evolution a tensor network.
Specifically, we now insert copies of the identity to rewrite \cref{eqn:observable_sampling} as
\begin{equation}\label{eqn:observable_sampling_mbl}
   \langle\tilde{O}\rangle_{t} = \bra{\psi(0)}\tilde{U}^\dag e^{i\tilde{H}_{\sigma}t} \tilde{U} O\tilde{U}^\dag e^{-i \tilde{H}_{\sigma}t}\tilde{U} \ket{\psi(0)},
\end{equation}
where $\tilde{H}_{\sigma} \equiv \tilde{U} \tilde{H} \tilde{U}^{\dag}$ (in words, $\tilde{H}_{\sigma}$ takes the form of~\cref{eqn:HMBL} but with $\sigma_{j}$ in place of $\tilde{\tau}_{j}$ and $\tilde{J}_{I}$ in place of $J_{I}$). 
We calculate these expectation values using a quantum circuit of the form in \cref{fig:H1_contract}.
We order the qubits going from bottom to top and evolution time from left to right.
Following the structure of \cref{eqn:observable_sampling_mbl}, the first section of the circuit applies $\tilde{U}$ to convert to the truncated LIOM basis. 
The second section evolves under the truncated Hamiltonian. 
After converting back to the original basis by using $\tilde{U}^{\dag}$, the operator $O$ is applied.
Then the previous steps are repeated in reverse.
Because the terms of $\tilde{H}_{\sigma}$ pairwise commute, we are allowed to choose the order in which each term appears.
Our choice is the following.
Place all evolution under terms supported on site 1 first, refer to these terms as $\tilde{H}_{1}$, and define $\tilde{V}_{1} \equiv e^{-i\tilde{H}_{1}t}$.
Then, place all evolution under terms supported on site 2, but not site 1, and refer to this as $\tilde{H}_{2}$.
Similarly, define $\tilde{V}_{2} \equiv e^{-i\tilde{H}_{2}t}$.
Continue in this way until all Hamiltonian evolution is accounted for. See \cref{fig:H1_contract} for a depiction of the circuit for $O = \sigma_{4}^{z}P_{3}P_{2}P_{1}$ and $N = 8$.
\begin{figure*}[t!]
    \centering
    \includegraphics[width=\textwidth]{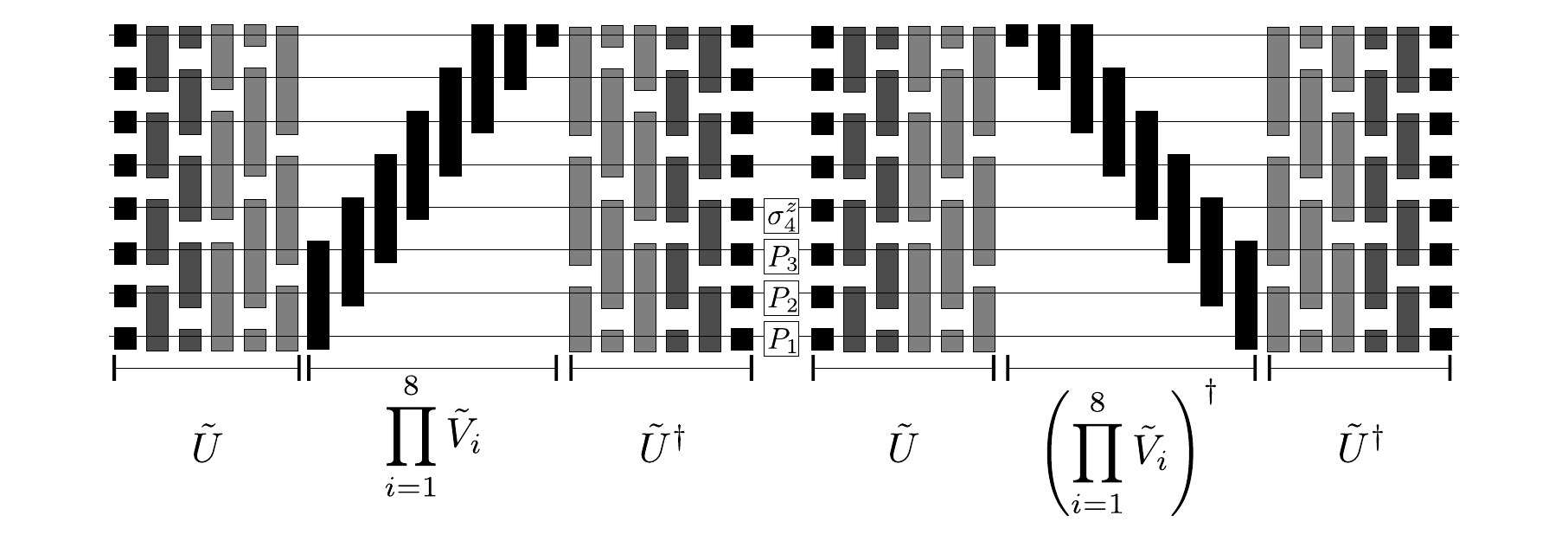}
    \caption{Example of the quantum circuit that calculates a relevant product of local observables $O$ on a lattice of $N = 8$ sites.
    Here $O = \sigma_{4}^{z}P_{3}P_{2}P_{1}$.
    }
    \label{fig:H1_contract}
\end{figure*}
Note that generating $\tilde{V}_{i}$ is an efficient process; there are at most $\binom{r_{J}}{k}$ $k$-site terms that involve site $i$ (but no site before $i$) and have physical range at most $r_{J}$.
Thus, there are at most $2^{r_{J}} \sim \mathrm{poly} N$ Hamiltonian evolution unitaries that must be multiplied together to generate each of the $N$ unitaries $\tilde{V}_{i}$. 
We treat each unitary in the evolution as a tensor, and we contract these tensors ``qubit-wise'' as opposed to ``time-wise.''
That is, instead of contracting tensors in the order that they appear in \cref{eqn:observable_sampling_mbl}, we first contract together every tensor that intersects qubit 1.
We then contract this much larger tensor with every other tensor that intersects qubit 2, and so forth.
Contracting the tensors ``time-wise'' would quickly lead us to an extensively sized tensor spanning some $\Theta(N)$ portion of the system, and evaluating a contraction involving this extensive tensor would take an exponentially long amount of time; contracting the tensors ``qubit-wise'' avoids this issue. 
Ensuring that our algorithm only ever produces tensors with $\mathcal{O}(\log N)$ legs would be sufficient to demonstrate a polynomial time algorithm.
This is because $\tilde{U}$ and $\tilde{U}^{\dag}$ each contain $\mathcal{O}(N\log N)$ constituents, $e^{- i \tilde {H} t}$ contains only $\mathcal{O}(N)$ terms (as we have decomposed it into $\{\tilde{V}_{i}\}$), and there are at most $N$ tensors coming from $O$.
Thus, the total number of tensors, and, correspondingly, the total number of legs that could be contracted, is only $\tilde{\mathcal{O}}(\mathrm{poly }N)$ (where the tilde indicates that we are ignoring logarithmic factors of $N$).
Therefore, the maximum amount of time this algorithm could take would be $\tilde{\mathcal{O}}(\mathrm{poly }N) \cdot 2^{\mathcal{O}(\log N)} = \tilde{\mathcal{O}}(\mathrm{poly }N)$. 

Unfortunately we can only guarantee that our algorithm produces tensors with $\mathcal{O}(\mathrm{polylog} N)$-many legs.
Intuitively, we cannot guarantee against an adversarial placement of constituents in $\tilde{U}, \tilde{U}^{\dag}$ whereby there is a jagged ``skyline'' of tensors leading to $\mathrm{polylog} N$ leftover legs after a qubit is contracted.
Repeating the above analysis means the algorithm can take as long as $\tilde{\mathcal{O}}(\mathrm{poly }N) \cdot 2^{\mathcal{O}(\mathrm{polylog} N)}$. 
This is not a polynomial-time algorithm; it is quasipolynomial, which means it is faster than any exponential-time algorithm, but slower than any polynomial-time algorithm.
Lemma $\ref{lem:contraction_algorithm_quasipolynomial}$ formalizes this rough argument. 

\begin{lemma}\label{lem:contraction_algorithm_quasipolynomial}
	Given the truncation of an MBL Hamiltonian and the quasilocal unitary that diagonalizes it, as in \cref{eq:trunc_J,eq:trunc_U}, following the qubit-wise contraction scheme never creates a tensor with more than $\O{[\log N]^{3}}$ leftover legs.
\end{lemma}
\begin{proof}
We will crudely upper-bound the total number of legs at any stage of the algorithm.
It is simple to see that the largest possible tensor occurs at the end of contracting all tensors intersecting a qubit $k$.
At this point consider a bound on the worst-case scenario where each of the $n$-site constituents in $\tilde{U}$ extends $n-1$ sites above qubit $k$, and $\tilde{V}_{k}$ extends $r_{J}-1$ sites above qubit $k$.
By naively ignoring that the internal legs should be contracted, it is straightforward to verify that this tensor possesses fewer than $4[2(2-1) + 3(3-1) + \cdots + r_{U}(r_{U}-1)) + 2(r_{J} -1) + 2] = \mathcal{O}([\log N]^{3})$ legs.
Because this is the worst-case scenario, the bound is thus proven. 
\end{proof}
\cref{lem:contraction_algorithm_quasipolynomial} bounds the size of any one tensor contracted in the algorithm, thus placing a quasipolynomial-time bound on any individual contraction. 
The total number of contraction operations is itself bounded by a polynomial in $N$.
Finally, we proved earlier that the distributions generated by $H$ and $\tilde{H}$ are $\epsilon$-close for polynomial evolution time.
Thus, the following theorem holds: 

\begin{theorem}\label{thm:easiness}
  For evolution time $t = \O{\mathrm{poly} N}$, the contraction algorithm takes time quasipolynomial in $N$, which means \cref{prob:mblsampling} can be solved in quasipolynomial time.
\end{theorem}

Additionally, observe that \cref{thm:easiness} can be extended to quasipolynomial evolution time with little effort.
Tracking the rest of the proof, we see that truncating the quasilocal unitary and the MBL couplings to length scales polylogarithmic in $N$ will make $\norm{\Delta H}$ small enough to counteract the larger evolution time $t$. 
A polylogarithmic truncation distance, however, does not change the quasipolynomial conclusion of \cref{lem:contraction_algorithm_quasipolynomial}. 
Finally, we note \cref{thm:easiness} holds in the worst case, meaning for any possible choice of coupling strengths and quasilocal unitary that obey our definition of MBL. 


\section{Hardness After Exponential Time}\label{sec:Hardness}
In contrast to the quasi-easiness result for strong simulation in \cref{sec:Easiness}, it is also possible to show, via a comparison to Instantaneous Quantum Polynomial (IQP) circuits \cite{Bremner2011}, that weak simulation of, or sampling from, MBL systems becomes formally hard on a classical computer after a time exponential in the system size.
\begin{theorem}\label{thm:hardness}
\cref{prob:mblsampling} is classically hard when the evolution time $t \geq \Omega(e^{N^{\delta}/\xi})$ for any $\delta > 0$.
\end{theorem}
\begin{proof}
For simplicity, we start with $\delta = 1/2$ and give a family of hard instances of the problem, described by the couplings $J_{i_{1}\dots i_{p}}$ in the $\tau$ basis and the quasilocal unitaries $U$ that satisfy our definition of MBL.
We rely on the hardness construction of Ref.~\cite{Bermejo-Vega2018}, which shows that evolution under a nearest-neighbor, commuting 2D Hamiltonian for constant time can be hard to classically simulate.
We implement the nearest-neighbor 2D dynamics using selective long-range interactions in 1D to generate an effective square grid of size $\sqrt{N}\times \sqrt{N}$, as depicted in \cref{fig:IQP_hardness}.
The 1D Hamiltonian $H_{1}$ is an MBL Hamiltonian of the form in \cref{eqn:HMBL} with coupling coefficients given by
\begin{align}\label{eqn:hardness_coeff}
J_{i_{1}} &= h_{i_{1}} = \mathcal{O}(1), \\
J_{i_{1}i_{2}} &= \! \begin{cases}
-e^{-\frac{\sqrt{N}}{\xi}} & i_2 -i_1 = 1, \; \, i_1 \neq 0 \bmod \sqrt{N}\\
-e^{-\frac{\sqrt{N}}{\xi}} & i_2-i_1 = \sqrt{N}
\end{cases},\\
J_{i_{1}\dots i_{p}} &= 0 \text{ if $p \geq 3$},
\end{align}
(where we have assumed, for simplicity, $\sqrt{N}$ is an integer) and l-bits given by
\begin{align}\label{eqn:LIOMs_hardness_1}
\tau^z_i &= \sigma^x_i, \\
\label{eqn:LIOMs_hardness_2}
\tau^x_i &= \sigma^z_i.
\end{align}
The Hamiltonian $H_1$ clearly satisfies our definition of a canonical MBL Hamiltonian; the coupling coefficients decay sufficiently quickly, and it is easy to verify that the Hadamard gate $U_{i}^{(1)} = \frac{1}{\sqrt{2}}\begin{pmatrix} 1 & 1 \\ 1 & -1 \end{pmatrix} = \mathrm{H}$ is unitary, satisfies \cref{eqn:frob_closeness} with $q = 4$, and effects \cref{eqn:LIOMs_hardness_1,eqn:LIOMs_hardness_2}.
It can be seen that (up to a local basis change $\sigma^z_i \leftrightarrow \sigma^x_i$) time-evolving $\ket{0}^{N}$ under $H_1$ for time $t = \pi e^{\frac{\sqrt{N}}{\xi}}/4$ is equivalent to time-evolving $\ket{+}^{N}$ (with $\ket{+}$ the +1 eigenstate of $\sigma^{x}$) under the 2D Hamiltonian $H = -\sum_{\avg{i,j}} \frac{\pi}{4}\sigma^z_i \sigma^z_j + \sum_i \frac{\pi}{4}e^{\frac{\sqrt{N}}{\xi}} h_i \sigma_i^z$ for time $1$, where $\avg{i, j}$ denotes neighboring sites.
If the local fields $h_i$ are chosen randomly such that $e^{\frac{\sqrt{N}}{\xi}} h_i \in \{ 1, 3/2 \} \bmod 4 $ \cite{LocalFields} with equal probability, evolution under $H_1$ on the initial state $\ket{0}^N$ implements Architecture I of Ref.~\cite{Bermejo-Vega2018}.
This Architecture is a Measurement-Based Quantum Computing (MBQC) scheme that is based on the hardness of IQP sampling.
Essentially, a disordered product state is prepared on a 2D grid after which controlled $\sigma^{z}$ gates are applied across each edge and a measurement in the $\sigma^{x}$ is performed. 
Sampling from the output distribution of this scheme is hard assuming two plausible complexity-theoretic conjectures (namely: the Polynomial Hierarchy is infinite and approximating partition functions of Ising models is average-case hard --- the original paper contained a third conjecture related to anticoncentration of certain classes of random circuits, but this conjecture was proven in a later work \cite{Hangleiter2018}). 
Therefore, for times $t = \Omega(e^{\sqrt{N}/\xi})$, \cref{prob:mblsampling} is hard, assuming certain plausible conjectures in computational complexity \cite{Aaronson2011,Bremner2016,Bouland2019,Bermejo-Vega2018}.

Recent work in Ref.~\cite{Maskara2019a} allows us to extend $\delta = 1/2$ to any $0 < \delta < 1$.
Because Architecture I of Ref.~\cite{Bermejo-Vega2018} may be implemented on any rectangular grid with non-constant dimensions, we may sculpt an effective 2D grid of size $N^{\delta} \times N^{1-\delta}$, where the long-range coefficients in \cref{eqn:hardness_coeff} now couple sites at a distance of only $N^{\delta}$.
The rest of our arguments go forward unchanged, except the time it takes to implement the architecture is now exponential in $N^{\delta}/\xi$.
\end{proof} 

\cref{thm:hardness} thus proves that there is a family of MBL Hamiltonians that are hard to classically simulate after an exponentially long evolution time.
Note that while it is hard to simulate this particular family of Hamiltonians in the average case, per the results in Ref.~\cite{Bermejo-Vega2018}, we observe that this family of Hamiltonians is itself somewhat fine-tuned. 
We therefore say that classically simulating MBL Hamiltonians for exponentially long evolution time is hard in the worst case.
However, \cref{thm:easiness} provided a quasipolynomial time algorithm to simulate MBL Hamiltonians for polynomially long evolution time, even in the worst case (as the results were indepdendent of the couplings and quasilocal unitary definining the Hamiltonian). 
Together, \cref{thm:easiness,thm:hardness} point toward a possible transition in the classical worst-case hardness of \cref{prob:mblsampling} between polynomial and exponential evolution times (and prove such a transition between logarithmic and exponential times for Hamiltonians that are short-range in the p-bit basis). 
Furthermore, \cref{thm:hardness} stands in stark contrast to the easiness result from Ref.~\cite{Deshpande2018Dynamical} that single-particle localized systems of bosons admit an efficient sampling algorithm for all evolution times.
However, it matches the intuition behind the hardness result in Ref.~\cite{Deshpande2018Dynamical}, where sampling free boson systems becomes difficult when the system is no longer approximately separable.
Similarly, \cref{prob:mblsampling} becomes provably hard when the system has evolved sufficiently for entanglement to spread across a distance scaling polynomially with $N$, where this long-range entanglement means the state of the system is no longer approximately separable \cite{kim_local_2014}.

\begin{figure}
\centering
\includegraphics[width=\linewidth]{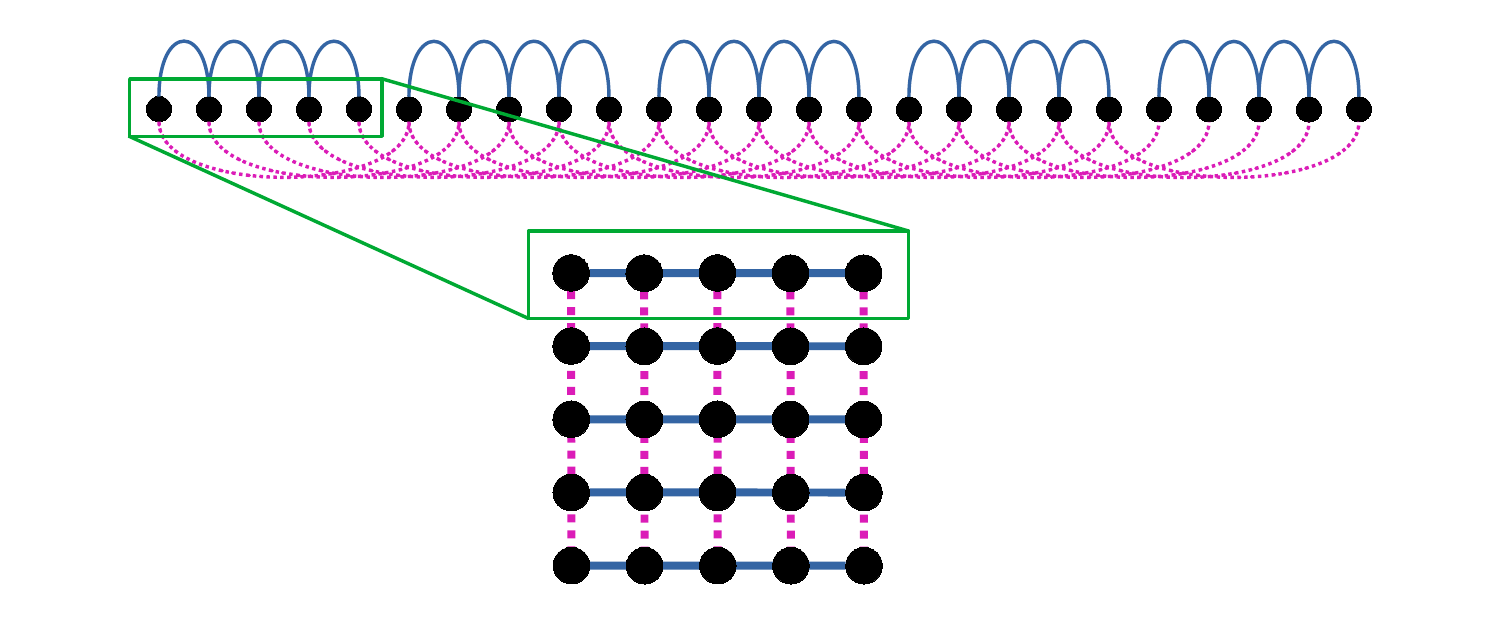}
\caption{Example illustrating the 1D-to-2D mapping of a Hamiltonian $H$ with coefficients given in \cref{eqn:hardness_coeff} acting on $N = 25$ qubits. 
The solid blue (dotted pink) lines depict the interactions with $\abs{i_{1}-i_{2}} = 1$ ($\abs{i_{1}-i_{2}} = \sqrt{N}$) in the true 1D lattice.
While the interactions differ in their locality, they have the same magnitude for simplicity in implementing the proposed architecture.
The single-site terms are not depicted.}
\label{fig:IQP_hardness}
\end{figure}

\section{Quantum complexity of simulating MBL systems}\label{sec:Quantum}
In this section, we focus on the quantum circuit complexity of approximately implementing the time-evolution operation $e^{-iHt}$ for an MBL Hamiltonian $H$.

\begin{definition}[Approximate circuit complexity]
The $\epsilon$-approximate circuit complexity $C_\epsilon$ of a unitary $U$ is the minimum circuit size $k$ of a circuit $G = G_k \ldots G_2 G_1$ composed of the standard gate set containing $\mathrm{CNOT}$, Hadamard, and $\pi/8$-phase gates ($\{\mathrm{CNOT}, \mathrm{H}, \mathrm{T}\}$) that approximates $U$ up to error $\epsilon$.
More formally, let 
\begin{align}
S_\epsilon(U) = \{G = G_k \ldots G_2 G_1 \text{ such that } \\ \nonumber
\norm{G-U} \leq \epsilon \ \text{and } G_i \in \{\mathrm{CNOT}, \mathrm{H}, \mathrm{T}\} \} 
\end{align}
be the set of all gate decompositions of $U$ over the standard gate set achieving error $\leq \epsilon$.
For a gate decomposition $G$, let $\abs{G} \equiv k$ denote its size.
Then
\begin{align}
C_\epsilon(U) \equiv \min_{G \in S_\epsilon(U)}{|G|}.
\end{align}
\end{definition}

We show that for evolution under MBL Hamiltonians, the complexity growth with respect to evolution time is slower than linear, which we denote through the symbol $o(t)$ \cite{asymptotic_notation} in the theorem below (while the gate complexity ultimately depends on the chosen gateset, the Solovay-Kitaev theorem ensures that this dependence is weak enough to not change this sublinear scaling).

\begin{theorem}[Sublinear growth of MBL circuit complexity]
For a Hamiltonian $H$ satisfying the criterion of MBL as defined in \cref{eqn:HMBL} and \cref{def:QLU} with $\xi < 1/(4\log 2)$, the approximate circuit complexity $C_\epsilon$ for constant $\epsilon$ obeys the bound
\begin{align}
C_\epsilon(e^{-iHt}) \leq \mathrm{poly}(N) \mathrm{polylog}(N^2t) \times o(t).
\end{align}
\end{theorem}
\begin{proof}
We leverage results from \cref{sec:Truncation}.
Our strategy to approximate the time-evolution unitary $e^{-iHt}$ is to apply instead the truncated evolution $e^{-i\tilde{H}t}$.
We have already argued that $\|e^{-iHt} - e^{-i\tilde{H}t}\| \leq \norm{\Delta H} t$, so, therefore, it suffices to choose $\tilde{H}$ so that $\norm{\Delta H} \leq \epsilon/t$.
In order to ensure that the unitary $e^{-i{H}t}$ can be applied with small circuit complexity, we make use of the fact that the (truncated) quasilocal unitary (approximately) diagonalizes the Hamiltonian:
\begin{align}
e^{-i\tilde{H}t} = \tilde{U}^\dag e^{-i\tilde{H}_\sigma t} \tilde{U}. \label{eqn:decomposition}
\end{align}
The cost of implementing the evolution under the MBL Hamiltonian comes from two parts: the first part stems from the cost of diagonalizing the Hamiltonian by implementing the quasilocal unitary $\tilde{U}$, and the second part comes from the complexity of applying time evolution under the truncated Hamiltonian in the physical basis, namely implementing $e^{-i\tilde{H}_\sigma t}$. This is the cost of implementing the last three sections (after the column of single-site observables) of the circuit depicted in \cref{fig:H1_contract}. 

The cost of applying $\tilde{U}$ can be upper bounded from the fact that it consists of gates that act on no more than $r_U = \Theta(\xi b \log N)$ many qubits at a time.
In the decomposition of $\tilde{U}$ as a quasilocal unitary, there are $N$ single-qubit unitaries, $2\ceil{N/2} = \mathcal{O}(N)$ two-qubit unitaries, and so on until the last layer of $\mathcal{O}(N)$ unitaries acting on $r_U$ qubits at a time.
Every unitary acting on $k$ qubits can be decomposed exactly into an $\mathcal{O}(k^2 2^{2k})$-long sequence of single-qubit and CNOT unitaries \cite{Nielsen2011}.
Using approximate synthesis algorithms over the Clifford+T gate set \cite{Kliuchnikov2015}, each of the single-qubit unitaries can be further decomposed into single-qubit gates from the standard gate set at only polylogarithmic overhead in the achieved error.
More precisely, the circuit complexity is upper bounded by 
\begin{align}
N \log(\delta^{-1}) + 4N \log(\delta^{-1})  \cdot 2^{2\cdot 2} + 9N \log(\delta^{-1}) \cdot 2^{2\cdot 3} + \ldots \nonumber
\\ + N r_U^2 \log(\delta^{-1}) \cdot 2^{2\cdot r_U}, \label{eqn:UCircuitComplexity}
\end{align}
where $\delta$ is the error made in approximating each local unitary.
The terms in \cref{eqn:UCircuitComplexity} correspond sequentially to the complexity of simulating the single-site, two-site, \dots, $r_{U}$-site terms. 
The first term does not contain the factor $2^{2k}$ because it corresponds to single-qubit unitaries.
The total error made in approximating $\tilde{U}$ then sums to
\begin{align}
& \delta \times (N + 4N\cdot 2^{2\cdot 2} + 9N \cdot 2^{2\cdot 3} + \ldots r_U^2 N \cdot 2^{2\cdot r_U})
\\ & \leq \delta N \times (1^2 \cdot 4^1 + 2^2 \cdot 4^2 + 3^2 \cdot 4^3 + \ldots {r_U}^2 \cdot 4^{r_U})
\\ & = N \delta \times \frac{4}{27} \left((9r_U^2 - 6r_U + 5)4^{r_U} -5\right)
\\ & \leq {2N \delta r_U^2 4^{r_U}},
\end{align}
which we set to be $\epsilon/6$ by choosing $\delta = \epsilon/(12Nr_U^2 4^{r_U})$.
Hence
\begin{align}
C_{\epsilon/6}(\tilde{U}) & \leq N \log(\delta^{-1}) \times (4 + 4\cdot 4^2 + 9 \cdot 4^3 + \ldots r_U^2 \cdot 4^{r_U})
\\ & = \mathcal{O}{(N \log(\delta^{-1}) r_U^2 4^{r_U})}
\\ & = \mathcal{O}\left(N 4^{r_U} r_U^2 \left(r_U \log(4) + \log\left(\frac{12Nr_U^2}{\epsilon}\right)\right) \right).
\end{align}

The cost of implementing $e^{-i\tilde{H}_\sigma t}$ can also similarly be upper bounded.
Here, for simplicity, we use the decomposition of $e^{-i\tilde{H}_{\sigma}t}$ from \cref{sec:Easiness}, where we combined unitaries acting on site $i$ (but not before $i$) into $\tilde{V}_{i}$.
This decomposition has $N$ unitaries of size at most $r_{J}$, meaning the gate complexity for  $e^{-i\tilde{H}_\sigma t}$ is upper bounded by $\mathcal{O}(N\log(\delta^{-1})r_J^2 4^{r_{J}})$, and the total error made in approximating these gates is thus $\mathcal{O}(N\delta r_J^2 4^{r_{J}})$. 
We again set this error equal to $\epsilon/6$ with a choice now of $\delta = \epsilon/(12 N r_J^2 4^{r_{J}})$, similarly yielding a gate complexity of 
\begin{align}
C_{\epsilon/6}(e^{-i\tilde{H}_\sigma t}) = \mathcal{O}\left(N 4^{r_J} r_J^2 \left(r_J \log(4) + \log\left(\frac{12Nr_J^2}{\epsilon}\right)\right) \right).
\end{align}

Combining everything, the total error for implementing the decomposition in \cref{eqn:decomposition} is $\epsilon/6 \times 3 = \epsilon/2$.
The total error in implementing $e^{-iHt}$ is thus upper bounded by the sum of the error in approximating $e^{-iHt}$ by $e^{-i\tilde{H}t}$ plus the error in decomposing $e^{-i\tilde{H}t}$ into a sequence of single and two-qubit gates:
\begin{align}
\epsilon/2 + \norm{\Delta H}t \leq \epsilon/2 + tC_{J}Nr_{J}e^{-kr_{J}} + tC_{U}N^{2}e^{-\frac{r_{U}}{2\xi}},
\end{align}
where we used \cref{eqn:error} to bound the second term.
We make the choices $r_J = (1.01) \log(Nt)/k$ and $r_U = 2.02 \xi \log(N^2t)$ so that the total error is at most
\begin{align}
& \epsilon/2 + C_J (Nt)^{-0.01} \log (Nt)/k + C_U (N^2t)^{-0.01} \nonumber
\\ & < \epsilon. 
\end{align}
With these choices, the total gate cost of simulating the entire circuit becomes $2C_{\epsilon/6}(\tilde{U}) + C_{\epsilon/6}(e^{-i\tilde{H}_\sigma t})$:
\begin{align}
C_\epsilon(e^{-iHt}) \leq \mathcal{O} &\left(N (N^2t)^{2.02 \xi \log 4} \mathrm{polylog}(N^2t) \right. \nonumber 
\\ & + N(Nt)^{1.01 \log 4/k}  \left. \mathrm{polylog}(Nt) \right).
\end{align}
As long as $\xi < 1/(2.02 \log 4) = 1/(4.04 \log 2)$, the exponent of $t$ in the first term is smaller than 1.
The same choice also ensures that the exponent of $t$ in the second term is smaller than 1 because $1.01 \log 4/k = 1.01 \log 4/(1/\xi - \log 2) < 2.02/3.04 < 1$.
\end{proof}
Thus, for sufficiently localized MBL Hamiltonians, the quantum circuit complexity is sublinear in time.
Such sublinear scaling contrasts MBL systems with chaotic Hamiltonians, which are conjectured to have quantum circuit complexity growing linearly with time, as supported by recent work in \cite{Brandao2021,Haferkamp2021a}. This provides a complexity-theoretic understanding of why MBL systems are unlikely to generate such chaotic dynamics.
This conclusion is intuitively consistent with the slow logarithmic spread of entanglement that is characteristic of MBL systems. 


\section{Conclusion and Outlook}\label{sec:Conclusion}
In this work, we have developed the best known formal results on the complexity of simulating MBL systems.
We have applied results in the literature to show that MBL systems evolved for time logarithmic in the system size admit an efficient classical strong simulaion, and, hence, sampling, algorithm. 
Further, we have demonstrated a quasipolynomial-time algorithm that can strongly simulate sufficiently localized MBL systems that have evolved for any (quasi)polynomially long time. 
While we have not quite provided a polynomial-time algorithm, the quasipolynomial-time algorithm is suggestive that possible improvements may lead to a formal proof of easiness. 
In particular, either the algorithm may be improved, potentially by leveraging the work on spectral tensor networks in Refs.~\cite{chandranSpectralTensorNetworks2015d, pollmannEfficientVariationalDiagonalization2016, wahlEfficientRepresentationFully2017} to make formal complexity statements in the case of quasilocal integrals of motion, or it may be possible to develop an algorithm that samples directly instead of going through the harder task of strong simulation. 
We leave these possible improvements (or the proof that they are impossible) as important open questions for future work. 
Furthermore, our proof holds only for Hamiltonians with LIOMs that are highly localized to a distance of about $\xi < 1/\log2$, in units of the lattice spacing. 
We do not consider this restriction to be too problematic, as previous work, e.g., Ref.~\cite{DeRoeck2017Stability}, has demonstrated that LIOMs may need to be highly localized for MBL systems to remain stable. 
It would be interesting, however, to understand more fully if this restriction is an artifact of our techniques, or if it is explained by some physical transition in MBL systems. 
Additionally, all of our results are based on bounding the worst-case scenario without explicitly accounting for disorder in our couplings, and studying the effect of disorder is an interesting open question.
Finally, it is also crucial to explore the easiness of simulating MBL systems when one only has access to $H$ in the p-bit basis. 

Apart from our easiness results, we have shown by a comparison to the problem of sampling from IQP circuits that a family of random MBL systems becomes hard to simulate after a time exponentially long in the system size. 
This family, while entirely consistent with our definition of MBL, is rather fine-tuned and likely has little overlap with the family of MBL Hamiltonians induced by disorder in the physical basis. 
Therefore, it would be quite valuable to determine in future work whether average-case hardness at exponential evolution times also holds for a more natural family of disorder-induced MBL Hamiltonians.

Additionally, we have also detailed the gate complexity of quantum simulation of MBL systems, and we have shown that for systems with localization length $\xi < 1/(4\log2)$, this gate complexity is sublinear. As for our results on classical simulation, it would be interesting to determine whether this localization length restriction is an artifact of our proof techniques or is physical. It would also be enlightening to investigate the connection between these results and the literature on fast-forwarding Hamiltonian evolution \cite{Atia2017}. 

Finally, so far we have specified entirely to MBL systems defined in 1D.
Indeed, there is significant debate over whether disorder-induced MBL can even exist in higher dimensions \cite{Abanin2019} (for example, the proof of MBL and LIOM structure in Ref.~\cite{imbrie_many-body_2016} relies crucially on the 1D nature of the system).
However, the natural generalization of our definition of MBL to higher dimensions would allow for MBL Hamiltonians that implement Architecture I of Ref.~\cite{Bermejo-Vega2018} directly (i.e., without sculpting an effective 2D grid using exponentially decaying interactions) in constant time. 
Thus, sampling from higher-dimensional MBL systems becomes hard very quickly, after evolution time $t = \mathcal{O}(1)$. However, other less natural extensions might exclude fast implementations of Architecture I, so the hardness of simulating higher-dimensional MBL systems still deserves further examination.
\begin{acknowledgments}
	We thank Eli Chertkov, Elizabeth Crosson, Bill Fefferman, James Garrison, Vedika Khemani, Nishad Maskara, Paraj Titum, Minh C. Tran, and Brayden Ware for helpful discussions.
	A.\,E., C.\,L.\,B., and A.\,V.\,G.\,acknowledge funding from the DoD, DoE ASCR Accelerated Research in Quantum Computing program (award No.~DE-SC0020312), NSF PFCQC program, AFOSR, DoE QSA, NSF QLCI (award No.~OMA-2120757), DoE ASCR Quantum Testbed Pathfinder program (award No.~DE-SC0019040), AFOSR MURI, U.S.~Department of Energy Award No.~DE-SC0019449, ARO MURI, and DARPA SAVaNT ADVENT. 
	A.\,D.\, acknowledges support from the National Science Foundation RAISE-TAQS 1839204 and Amazon Web Services, AWS Quantum Program.
	This research was performed in part while C.\,L.\,B.~held an NRC Research Associateship award at the National Institute of Standards and Technology.
	The Institute for Quantum Information and Matter is an NSF Physics Frontiers Center PHY-1733907. D.\,A. acknowledges support from the Swiss National Science Foundation and from the European Research Council (ERC) under the European Union's Horizon 2020 research and innovation program (grant agreement No. 864597).
	
\end{acknowledgments}
\bibliographystyle{apsrev4-2}
\bibliography{arxiv_v1.bib}

\begin{thebibliography}{48}
\bibitem [{asy()}]{asymptotic_notation}%
  \BibitemOpen
  \href@noop {} {}\bibinfo {note} {We say that $f = \mathcal{O}(g)$ if $f/g
  \nrightarrow \infty$ as $n \to \infty$, and $f = \Omega(g) \iff g =
  \mathcal{O}(f)$. Similarly, $f = o(g)$ means $f/g \to 0$ as $n \to \infty$,
  and $f = \omega(g) \iff g = o(f)$. Finally, if $f = \mathcal{O}(g)$ and $g =
  \mathcal{O}(f)$, we say $f = \Theta(g)$ (and $g = \Theta(f)$). The precise
  asymptotic dependence on $n$ can be arbitrary. Additionally, a tilde over the
  asymptotic symbol, such as $\tilde{\mathcal{O}}(g)$, means that we are
  ignoring logarithmic factors in $g$.}\BibitemShut {Stop}%
\bibitem [{fam()}]{families}%
  \BibitemOpen
  \href@noop {} {}\bibinfo {note} {To be precise, our definition is for a
  family of quasilocal unitaries $U$ with respect to the size $N$ of the
  system, otherwise all unitaries $U$ would be quasilocal given a sufficiently
  large, but still constant, $q$.}\BibitemShut {Stop}%
\bibitem [{Loc()}]{LocalFields}%
  \BibitemOpen
  \href@noop {} {}\bibinfo {note} {We could choose $h_{i}$ such that this
  relationship holds exactly, i.e. without working modulo 4, but that would
  require on-site terms that are exponentially small in the system
  size.}\BibitemShut {Stop}%
\end{thebibliography}

\newpage
\onecolumngrid
\let\oldaddcontentsline\addcontentsline
\renewcommand{\addcontentsline}[3]{}
\section*{Supplemental Material for: Simulation Complexity of Many-Body Localized Systems}
\let\addcontentsline\oldaddcontentsline
\setcounter{section}{0}
\setcounter{theorem}{0}
\setcounter{equation}{0}
\setcounter{lemma}{0}
\setcounter{definition}{0}
\setcounter{corollary}{0}
\renewcommand{\thetheorem}{S.\arabic{theorem}}
\renewcommand{\theequation}{S.\arabic{equation}}
\renewcommand{\thesection}{S.\Roman{section}}
\renewcommand{\thecorollary}{S.\arabic{corollary}}
\renewcommand{\thelemma}{S.\arabic{lemma}}

In this Supplemental Material we provide more details for the algorithm simulating MBL Hamiltonians evolved only for times at most logarithmic in the system size (\cref{sec:supp_log}), and we give mathematical proofs of \cref{eq:spn0,eqn:error,eq:rUrJ} deferred from the main text for clarity (\cref{sec:math}).

\section{Logarithmic Time Simulation}\label{sec:supp_log}
In this section, we give more details for a strong simulation algorithm for MBL Hamiltonians evolved for at most logarithmic times. 
As discussed in the main text, if the Hamiltonian $H$ is finite-range in the physical basis, Ref.~\cite{Osborne2006} provides an efficient representation of the propagator $e^{-iHt}$ for evolution time logarithmic in the system size $N$:
\begin{theorem}\label{thm:Osborne}[Ref.~\cite{Osborne2006}]
Assuming $H$ is finite-range in the physical basis, then one can construct an approximation $\tilde{U}$ to the propagator $U = e^{-iHt}$ such that $\norm{U-\tilde{U}} \leq \epsilon$ and $\tilde{U}$ may be computed with classical resources that are polynomial in $N$ and $1/\epsilon$ and exponential in $\abs{t}$. 
\end{theorem}
We have that for some initial state $\ket{\phi}$, $\norm{U\ket{\phi}-\tilde{U}\ket{\phi}}_{2} = \norm{U-\tilde{U}} \leq \epsilon$. Thus, approximate simulation of $U\ket{\phi}$ can be solved by exactly simulating $\tilde{U}\ket{\phi}$. As constructed in Ref.~\cite{Osborne2006}, $\tilde{U}$ is described in the matrix product operator formalism, which means that we have an algorithm that solves the problem of strong simulation for evolution of a product state under $\tilde{U}$.
This is because products of local observables admit a trivial matrix product operator formulation (as there is no correlation between the operators, a product of local observables is a matrix product operator with zero bond dimension). 
Because multiplication between reasonably sized matrix product operators is efficient, it is possibly to efficiently evaluate $\avg{e^{it\tilde{H}} (\prod_{i}O_{i})e^{-it\tilde{H}}}$.
As described in the main text, this also implies a sampling algorithm from the approximate distribution generated by measuring the initial state evolved under $\tilde{H}$ for time $t$:
\begin{corollary}
Provided $H$ is finite range in the physical basis, \cref{prob:mblsampling} is easy for $t = \O{\log N}$. 
\end{corollary}

The assumption that $H$ is finite-range in the physical basis is a technical one, but one that is reasonable, as many physical systems that are candidates for MBL, such as the disordered, short-range Ising model, fulfill such restrictions. 
Note, however, that finite-range Hamiltonians can also describe thermalizing systems. Thus, this result importantly establishes that there is a regime in which (many classes of) MBL systems admit sampling algorithms, but it does not use any of the salient features of MBL in order to distinguish it from the thermalizing phase.


\section{Mathematical Details}\label{sec:math}
Here we will present mathematical details deferred from the main text for clarity. 
\cref{lem:LIOM_closeness} bounds the difference between the full and approximate LIOMs discussed in the main text. 
\cref{lem:sum_bound_full} places a bound on the sum $S_{p,n_{0}}$. 
\cref{lem:inc_gamma_bound} \cite{Borwein_2007_Uniform} provides an intermediate result regarding the incomplete Gamma function that is useful in proving the bound on $S_{p,n_{0}}$. 
\cref{lem:true_truncated_difference} applies \cref{lem:LIOM_closeness} and \cref{lem:sum_bound_full} in order to bound the operator norm of the difference between the full and truncated Hamiltonians.

\begin{lemma}\label{lem:LIOM_closeness}
	Let $H$ be an MBL Hamiltonian with localization length $\xi < 1/\log 2$. 
	Let $U$ be a quasilocal unitary with localization length $\xi$ as in \Cref{def:QLU} such that $U$ diagonalizes $H$, and let $\tilde{U}$ be $U$'s truncation to constituents of range less than or equal to $r_{U} = 2a \xi \log N$ for some constant $a > 1$. 
	Finally, let $\tau_{i}^{\alpha} = U\sigma_{i}^{\alpha}U^{\dag}$ and $\tilde{\tau}_{i}^{\alpha} = \tilde{U}\sigma_{i}^{\alpha}\tilde{U}^{\dag}$. For large enough system sizes $N$, it follows that
	\begin{equation}
		\norm{\tau_{i}^{z} - \tilde{\tau}_{i}^{z}} \leq 8\sqrt{q}Ne^{-\frac{r_{U}}{2\xi}},
	\end{equation}
	where $\norm{\cdot}$ is the operator norm. 
\end{lemma}
\begin{proof}
Let $U = U'\tilde{U}$, where
\begin{align}
	\tilde{U} &= \prod_{n=1}^{r_{U}}\prod_{j=1}^{n}\prod_{i=0}^{\lfloor (N-n)/n \rfloor}U_{in + j}^{(n)}, \\
	U' &= \prod_{n=r_{U}+1}^{N}\prod_{j=1}^{n}\prod_{i=0}^{\lfloor (N-n)/n \rfloor}U_{in + j}^{(n)}.
\end{align}
 
Write $U_{in + j}^{(n)} = \id + \Delta_{in + j}^{(n)}$, where we use $\id$ to denote the identity operator on the appropriate Hilbert space.
\Cref{def:QLU} tells us that $\norm{\Delta_{in+j}^{(n)}} < \sqrt{q}e^{-\frac{n-1}{2\xi}}$. 
Also write $\tilde{U} = \id + \tilde{\Delta}$ and, similarly, $U' = \id + \Delta'$ such that:
\begin{align}
	\norm{\Delta'} &= \norm{\prod_{n = r_{U}+1}^{N}\prod_{j=1}^{n}\prod_{i = 0}^{\lfloor (N-n)/n \rfloor}(\id + \Delta_{in + j}^{(n)})-\id} \label{EQN:DeltaGreaterDef}\\
	\norm{\tilde{\Delta}} &= \norm{\prod_{n = 1}^{r_{U}}\prod_{j=1}^{n}\prod_{i=0}^{\lfloor (N-n)/n \rfloor}(\id + \Delta_{in + j}^{(n)})-\id} \label{EQN:DeltaLessDef}.
\end{align}
We now have that
\begin{align}
	\norm{\tau_{i}^{z} - \tilde{\tau}_{i}^{z}} &= \norm{U'\tilde{\tau}_{i}^{z}(U')^\dag - \tilde{\tau}_{i}^{z}} \\
	&= \norm{\Delta'\tilde{\tau}_{i}^{z} + \tilde{\tau}_{i}^{z}(\Delta')^{\dag} + \Delta'\tilde{\tau}_{i}^{z}(\Delta')^{\dag}} \\
	&\leq \norm{\Delta'} + \norm{(\Delta')^{\dag}} + \norm{\Delta'}\norm{(\Delta')^{\dag}} \label{EQN:Tau-Tau_Delta_Bound}.
\end{align}
Define a multi-index parameter $\alpha = (i, j, n) = (k, n)$ where $k = in + j$ specifies the left-most site of an $n$-site unitary. 
We may then rewrite \cref{EQN:DeltaGreaterDef}:
\begin{equation}
	\Delta' = \left[\prod_{\alpha}(1+ \Delta_{\alpha})\right] -1 = \left[\sum_{S}\prod_{\alpha \in S}(\Delta_{\alpha})\right] -1 = \sum_{S \neq \emptyset}\prod_{\alpha \in S}\Delta_{\alpha},
\end{equation}
where $S$ is a subset of the possible $\alpha$ indices. 
The triangle inequality and submultiplicativity yield
\begin{equation}
	\norm{\Delta'} < 
	\left[\sum_{S}\prod_{\alpha \in S}\left(\sqrt{q}e^{-\frac{n_{\alpha}-1}{2\xi}}\right)\right] - 1 = \left[\sum_{S}\prod_{\alpha}\left(e^{-[\frac{(n_{\alpha}-1)}{2\xi}-\frac{\log q}{2}]\mathbb{I}(\alpha \in S)}\right)\right] - 1,
\end{equation}
where the indicator $\mathbb{I}(x)$ is 1 (0) if $x$ is true (false), and $n_{\alpha}$ is the size of the unitary indexed by $\alpha$.
To evaluate this, we switch the sum and product.
In particular, instead of using the indicator and summing over subsets $S$, we can instead view the sum as a sum over all $\alpha$ where $n$ can take either the value $0$ or $n_{\alpha}$.
Define $A$ to be the number of possible $\alpha$.
Then
\begin{align}
	\left[\sum_{S}\prod_{\alpha}\left(e^{-[\frac{(n_{\alpha}-1)}{2\xi}-\frac{\log q}{2}]\mathbb{I}(\alpha \in S)}\right)\right] - 1 &= \left[\sum_{\alpha_{1} = \{0, n_{1}-1\}} \cdots \sum_{\alpha_{A} = \{0, n_{A}-1\}} \left(e^{-\frac{\alpha_{1}}{2\xi} + \frac{\alpha_{1}\log q}{2(n_{1}-1)}}\cdots e^{-\frac{\alpha_{A}}{2\xi} + \frac{\alpha_{A}\log q}{2(n_{A}-1)}}\right)\right]- 1 \\
	&= \left[\prod_{\alpha}\sum_{n = \{0, n_{\alpha}-1\}}\left(e^{-\frac{n}{2\xi}+\frac{n\log q}{2(n_{\alpha}-1)}}\right)\right] - 1 \\
	&= \left[\prod_{\alpha}\left(1 + \sqrt{q}e^{-\frac{n_{\alpha}-1}{2\xi}}\right)\right] - 1.
\end{align}

We rewrite the infinite product as the exponential of an infinite sum:
\begin{equation}
	\prod_{\alpha}\left(1 + \sqrt{q}e^{-\frac{n_{\alpha}-1}{2\xi}}\right) - 1 = e^{\sum_{\alpha}\log\left(1 + \sqrt{q}e^{-\frac{n_{\alpha}-1}{2\xi}}\right)} - 1. \label{EQN:ExpSumLog}
\end{equation}
We now examine the sum:
\begin{equation}
	\sum_{\alpha}\log\left(1 + \sqrt{q}e^{-\frac{n_{\alpha}-1}{2\xi}}\right) = \sum_{n > r_{U}}\sum_{k}\log\left(1 + \sqrt{q}e^{-\frac{n-1}{2\xi}}\right).
\end{equation}
For any given $n$ (which labels the number of sites on which the block acts nontrivially), there are $N-n+1$ possible unitaries (the left-most site can be any besides the last $n-1$). 
We trivially upper bound this by $N$ such that
\begin{equation}
	\sum_{n > r_{U}}\sum_{k}\log\left(1 + \sqrt{q}e^{-\frac{n-1}{2\xi}}\right) < \sum_{n > r_{U}}N\sqrt{q}e^{-\frac{n-1}{2\xi}} = \frac{N\sqrt{q}}{1-e^{-\frac{1}{2\xi}}}e^{-\frac{r_{U}}{2\xi}}.
\end{equation}
Let $r_{U} = 2a \xi \log N$ for some $a > 1$.
Plugging back into \cref{EQN:ExpSumLog} yields
\begin{equation}
	\norm{\Delta'} < \exp\left(\frac{\sqrt{q}}{1-e^{-\frac{1}{2\xi}}}N^{1-a}\right) -1 < \sqrt{q}\frac{(1+o(1))}{1-\frac{1}{\sqrt{2}}}N^{1-a}
\end{equation}
for large enough $N$.
In the above, we have used $\xi < 1/\log 2$.
Plugging this result back into \cref{EQN:Tau-Tau_Delta_Bound} yields the result:
\begin{equation}
	\norm{\tau_{i}^{z} - \tilde{\tau}_{i}^{z}} \leq 8\sqrt{q} N^{1-a} = 8\sqrt{q}Ne^{-\frac{r_{U}}{2\xi}} 
\end{equation}
for large enough $N$.
\end{proof}

\begin{lemma}\label{lem:sum_bound_full}
Assuming $\xi < \frac{1}{\log 2}$, we may prove two bounds.
First
	\begin{equation}
		S_{p,n_{0}} = \sum_{n = n_{0}}^{\infty}\binom{n}{p}e^{-\frac{n}{\xi}} \leq C
		\begin{cases}
			e^{-\frac{n_{0}}{\xi}} & p = 0 \\
			pe^{-ap} & n_{0} < n_{*}, p > 0 \\
			\frac{n_{0}^{p+1}\sqrt{p}}{p!}e^{-\frac{n_{0}}{\xi}} & n_{0} \geq n_{*}, p>0
		\end{cases},
	\end{equation}
where $a \equiv \log (e^{1/\xi}-1)$, $n_{*} \equiv p\frac{e^{1/\xi}}{e^{1/\xi}-1} = p(1-e^{-1/\xi})^{-1}$, and $C = 10.8$. 

And, for $0 \leq x_{1} \leq x_{2} \leq n_{0}$:
\begin{equation}\label{eqn:Spn_trivial}
		\sum_{p = x_{1}}^{x_{2}}S_{p,n_{0}} = \sum_{p = x_{1}}^{x_{2}}\sum_{n = n_{0}}^{\infty}\binom{n}{p}e^{-\frac{n}{\xi}} \leq \frac{1}{1-e^{-\kappa}}e^{-\kappa n_{0}},
\end{equation}
where $\kappa =  \frac{1}{\xi} - \log 2$.

\end{lemma}

\begin{proof}
The proof of the second bound is straightforward. 
We simply upper bound the sum over $p$ of $\binom{n}{p}$ as $2^{n}$. 
We then have that 
\begin{equation}
	\sum_{p = x_{1}}^{x_{2}}S_{p,n_{0}} \leq \sum_{n = n_{0}}^{\infty} e^{n(-1/\xi + \log 2)},
\end{equation}
from which the result follows from exactly summing the geometric series, which converges as long as $\xi < \frac{1}{\log 2}$.
We now move on to the more complicated case that retains the $p$-dependence.

\emph{Case 1 ($p = 0$):}
The $p=0$ case is a straightforward geometric series and the constant out front can be chosen to be anything greater than $\frac{1}{1-e^{-1/\xi}} < 2$ (as $\xi < \frac{1}{\log 2}$). 

\emph{Case 2 ($n_{0} < n_{*}$):}
We begin with Stirling's Approximation, which says that:
\begin{equation}
	\sqrt{\frac{2\pi}{e^{4}}}\sqrt{\frac{n}{p(n-p)}}\frac{n^{n}}{p^{p}(n-p)^{n-p}}\leq\binom{n}{p}\leq\frac{e}{2\pi}\sqrt{\frac{n}{p(n-p)}}\frac{n^{n}}{p^{p}(n-p)^{n-p}}. 
\end{equation}
Applying the upper bound we see that
\begin{equation}
	\sum_{n = n_{0}}^{\infty}\binom{n}{p}e^{-\frac{n}{\xi}} \leq \sum_{n = n_{0}}^{\infty}\frac{e}{2\pi}\sqrt{\frac{n}{p(n-p)}}\frac{n^{n}}{p^{p}(n-p)^{n-p}}e^{-\frac{n}{\xi}}.
\end{equation}
We now note that for $n \geq p+1 > 1$, $\sqrt{\frac{n}{p(n-p)}} \leq \sqrt{2}$ such that $\frac{e}{2\pi}\sqrt{\frac{n}{p(n-p)}} \leq 1$. 
Then, for $n_{0} \geq p+1$,
\begin{equation}\label{eqn:Sp_exp_bound}
	S_{p,n_{0}} \leq \sum_{n = n_{0}}^{\infty}\frac{n^{n}}{p^{p}(n-p)^{n-p}}e^{-\frac{n}{\xi}} = \sum_{n = n_{0}}^{\infty}e^{-\frac{n}{\xi} + n\log n - (n-p)\log(n-p) - p\log p}\equiv \sum_{n = n_{0}}^{\infty}e^{g(n)}. 
\end{equation}
We can eliminate the $n \geq p+1$ assumption by realizing that the final bound in \cref{eqn:Sp_exp_bound} still holds trivially when $n = p$, as the logarithmic terms in $g(n)$ vanish.
Thus, \cref{eqn:Sp_exp_bound} is valid for all pairs $n_{0} \geq p$, which we assume in order to make the combinatorial factor $\binom{n}{p}$ well-defined.

Maximizing the summand means maximizing $g(n)$, so we calculate:
\begin{align}
	\frac{\partial g}{\partial n} &= -\frac{1}{\xi} + \log n - \log (n-p), \\
	\frac{\partial^{2} g}{\partial n^{2}} &= \frac{1}{n} - \frac{1}{n-p} .
\end{align}
It is straightforward to calculate 
\begin{equation}
	g'(n)
	\begin{cases}
		= 0  &n = n_{*} = p\frac{e^{1/\xi}}{e^{1/\xi}-1} \\
		> 0  &n < n_{*} \\
		< 0  &n > n_{*} \\
	\end{cases}.
\end{equation} 
Furthermore, it is also straightforward to verify $g''(n) < 0$ for all $n > p$.
Thus, we see that $g(n)$, and hence $e^{g(n)}$, has a single maximum on $[n_{0}, \infty)$; it is at $n_{*}$ for $n_{0} < n_{*}$ and $n_{0}$ for $n_{0} \geq n_{*}$. 
Additionally, it will be useful to calculate that $g(n_{*}) = -ap$, where $a \equiv \log(e^{1/\xi}-1)$. 

We now bound the final sum in \cref{eqn:Sp_exp_bound} with an integral using a Riemann approximation. 
In particular, let $n_{*}^{-} = \floor{n_{*}}$ and $n_{*}^{+} = n_{*}^{-} + 1$. 
Then
\begin{align}
	\sum_{n = n_{0}}^{\infty} e^{g(n)} &= e^{g(n_{*}^{-})} +e^{g(n_{*}^{+})} + \sum_{n=n_{0}}^{n_{*}^{-}-1} e^{g(n)} + \sum_{n=n_{*}^{+}+1}^{\infty} e^{g(n)} \\
	&\leq 2e^{g(n_{*})} + \int_{n_{0}}^{n_{*}^{-}}e^{g(n)} dn + \int_{n_{*}^{+}}^{\infty}e^{g(n)} dn  \\
	&\leq 2e^{-ap} + \int_{n_{0}}^{n_{*}}e^{g(n)}dn + \int_{n_{*}}^{\infty}e^{g(n)}dn \\
	&=2e^{-ap} + \int_{n_{0}}^{n_{*}}e^{g(n)}dn + \int_{n_{*}}^{2n_{*}}e^{g(n)}dn + \int_{2n_{*}}^{\infty}e^{g(n)}dn \\
	&\equiv 2e^{-ap} + I_{<} + I_{<>} + I_{>} \label{eqn:e^g(n)_bound_separated},
\end{align}

Consider first $I_{<}$. 
There we can start by using that $g(n_{*})$ is maximal to make the trivial bound:
\begin{equation}
	I_{<} \leq (n_{*}-n_{0})e^{g(n_{*})} \leq pe^{-a(p+1)},
\end{equation}
where we have used that $n_{*}-n_{0} \leq n_{*}-p = pe^{-a}$.
Similarly, for $I_{<>}$, we may say that
\begin{equation}
	I_{<>} \leq n_{*} e^{g(n_{*})} \leq 2pe^{-ap},
\end{equation}
where we have used that $p < n_{*} < 2p$ because $ 1/\xi > \log 2$.

To bound $I_{>}$, we first invert the Stirling approximation from earlier and write:
\begin{equation}\label{eqn:stirling_invert}
	e^{g(n)} = e^{-\frac{n}{\xi}}\frac{n^{n}}{p^{p}(n-p)^{(n-p)}} \leq \frac{e^{2}}{\sqrt{2\pi}}\sqrt{\frac{p(n-p)}{n}}\binom{n}{p}e^{-\frac{n}{\xi}} \leq 3\sqrt{p}\binom{n}{p}e^{-\frac{n}{\xi}} \leq 3\sqrt{p}\frac{n^{p}}{p!}e^{-\frac{n}{\xi}}.
\end{equation}
We can thus bound
\begin{equation}
	I_{>} \leq 3\frac{\sqrt{p}}{p!}\int_{2n_{*}}^{\infty}e^{-\frac{n}{\xi}}n^{p}dn.
\end{equation}
Substituting $u = \frac{n}{\xi}$ and defining $u_{*} =  \frac{n_{*}}{\xi}$ yield
\begin{equation}
	I_{>} \leq 3\frac{\xi^{p+1}\sqrt{p}}{p!}\int_{2u_{*}}^{\infty}e^{-u}u^{p}du = 3\frac{\xi^{p+1}\sqrt{p}}{\Gamma(p+1)}\Gamma(p+1,2u_{*}),
\end{equation}
where $\Gamma(a)$ and $\Gamma(a,z)$ are the standard Gamma and Incomplete Gamma functions, respectively. 
We can bound the Incomplete Gamma Function using Lemma \ref{lem:inc_gamma_bound} provided $2u_{*} > p$, i.e. $2\frac{e^{1/\xi}}{e^{1/\xi}-1} > \xi$.
We can actually do better and show that $u_{*} > p$. 
Defining $x = 1/\xi$, we want to show $xe^{x} - e^{x} +1 > 0$ for $x \in (0, \infty)$.
At $x = 0$, the LHS is 0.
Taking a derivative of the LHS with respect to $x$ yields $xe^{x} > 0$ for $x \in (0, \infty)$.
Thus, the LHS is $0$ at $x=0$ and increasing, which means the inequality holds. 
With that in mind, we apply Lemma \ref{lem:inc_gamma_bound}:
\begin{equation}
	I_{>} \leq 3\frac{\xi^{p+1}\sqrt{p}}{p!}\int_{2u_{*}}^{\infty}e^{-u}u^{p}du \leq 3\frac{\xi^{p+1}\sqrt{p}}{p!}\frac{(2u_{*})^{p+1}e^{-2u_{*}}}{2u_{*}-p} = 3 (2^{p+1})\frac{n_{*}^{p+1}e^{-\frac{2n_{*}}{\xi}}}{\sqrt{p}p!}\frac{1}{2\frac{1}{\xi}(\frac{e^{1/\xi}}{e^{1/\xi}-1})-1}.
\end{equation}
Note that
\begin{multline}\label{eqn:bound_I>}
2^{p+1}\frac{n_{*}^{p+1}e^{-\frac{2n_{*}}{\xi}}}{\sqrt{p}p!} = \frac{(2p)^{p+1}}{\sqrt{p}p!}\left(\frac{e^{1/\xi}}{e^{1/\xi}-1}\right)^{p+1}e^{-2\frac{p}{\xi}\frac{e^{1/\xi}}{e^{1/\xi}-1}} \leq \frac{2}{\sqrt{2\pi}}e^{p(1+\log 2)}e^{-a(p+1)}e^{\frac{p + 1}{\xi}}e^{-2\frac{p}{\xi}\frac{e^{1/\xi}}{e^{1/\xi}-1}}\\
 \leq \sqrt{\frac{2}{\pi}}e^{-ap} \underbrace{e^{-a + p(1+\log 2) + \frac{p+1}{\xi} - \frac{2p}{\xi}\frac{e^{1/\xi}}{e^{1/\xi} - 1}}}_{\leq e^{\log 2}} \leq \sqrt{\frac{8}{\pi}}e^{-ap}.
\end{multline}
The last bound is rather involved, so we will explain the steps carefully.
We want to show that 
\begin{align}
	-a + p(1+\log 2) + \frac{p+1}{\xi} - \frac{2p}{\xi}\frac{e^{1/\xi}}{e^{1/\xi}- 1} = \underbrace{p(1+\log 2) + \frac{p}{\xi} - \frac{2p}{\xi}\frac{e^{1/\xi}}{e^{1/\xi} - 1}}_{\text{(A)}} + \underbrace{\frac{1}{\xi} - a}_{\text{(B)}} < \log 2.
\end{align}
We can show (A) $< 0$ using a strategy similar to when we proved that our bound on the incomplete gamma function was valid.
In particular, first note that we can effectively cancel $\frac{p}{\xi}$ with one factor of $\frac{p}{\xi}\frac{e^{1/\xi}}{e^{1/\xi - 1}}$ given that $\xi < \log 2$.
We then want to show that $p(1 + \log 2) - px \frac{e^{x}}{e^{x}-1} < 0$, where, again, $x = \frac{1}{\xi}$. 
Equivalently, we want to show that $xe^{x}-(1+\log 2)e^{x} + (1+\log 2) > 0$.
Again, the LHS is 0 at $x = 0$. 
And, again, taking a derivative of the LHS gives us $xe^{x} + e^{x} - (1+\log 2)e^{x} = xe^{x} - \log 2 e^{x}$, which is greater than 0 as long as $x > \log 2$, or $\xi < \frac{1}{\log 2}$. 
We then want to bound (B), and this is done by noting that in the limit that $\xi$ is very small, then $a = \log(e^{1/\xi}-1) \sim 1/\xi$ such that (B) $\sim 0$.
In fact, the maximum of (B) is simply $\log 2$, which occurs for $\xi = \frac{1}{\log 2}$. 
With all of that handled, we can then say that the final bound in \cref{eqn:bound_I>} is exponentially decreasing with $p$ only if $a >0$, which corresponds to $\xi < \frac{1}{\log 2}$ or $1/\xi > \log 2$.

We need to combine the bounds on all of the components of \cref{eqn:e^g(n)_bound_separated}:
\begin{align}
	2e^{-ap} + I_{<} + I_{<>} + I_{>} &\leq \left(2 + pe^{-a} + 2p + 3\sqrt{\frac{8}{\pi}}\frac{1}{2\frac{1}{\xi}(\frac{e^{1/\xi}}{e^{1/\xi}-1})-1} \right)e^{-ap} \\
	&\leq \left(\frac{2}{p} + e^{-a} + 2 + \frac{3}{p}\sqrt{\frac{8}{\pi}}\frac{1}{4\log2 - 1}\right)pe^{-ap} \\
	&\leq C_{1}pe^{-ap},
\end{align}
where we have used the fact that $1/\xi > \log 2$ and defined
\begin{equation}
	C_{1} = 5 + 3\sqrt{\frac{8}{\pi}}\frac{1}{4\log2 - 1} < 7.8.
\end{equation}
\emph{Case 3 ($n_{0} \geq n_{*}$):} For sufficiently small $\xi$, we have $n_{*}\sim p$. Assuming $n_{0} \geq p+1$, this means that $n_{0} > n_{*}$.
However, given that situation, we can use that $g(n)$ is decreasing after $n_{0}$ to go immediately from \cref{eqn:Sp_exp_bound} to
\begin{equation}\label{eqn:nmax_less_n0_plus_extra_term}
	\sum_{n = n_{0}}^{\infty} e^{g(n)} \leq e^{g(n_{0})} + \int_{n_{0}}^{\infty}e^{g(n)}dn.
\end{equation}

In comparison with the case where $n_{*} < n_{0}$, the integral $I_{<}$ effectively does not exist here, and the bound in $I_{>}$ comes from simply replacing $2n_{*}$ with $n_{0}$ (which is now the maximal contribution) and adding on the extra term in \cref{eqn:nmax_less_n0_plus_extra_term}.
First, using steps nearly identical to those above (noting in particular that \cref{lem:inc_gamma_bound} is valid because $n_{0} > n_{*} > p\xi$ by the earlier proof), we can bound 
\begin{equation}
	I_{>} \leq \frac{3}{2\log 2-1}\frac{n_{0}^{p+1}e^{-\frac{n_{0}}{\xi}}}{\sqrt{p}p!}.
\end{equation}
Then, the contribution from $e^{g(n_{0})}$ may be bounded by inverting Stirling's approximation as in \cref{eqn:stirling_invert}:
\begin{equation}
	e^{g(n_{0})} \leq 3\frac{n_{0}^{p}\sqrt{p}}{p!}e^{-\frac{n_{0}}{\xi}}.
\end{equation}
Combining the two yields 
\begin{align}
	e^{g(n_{0})}  + I_{>} &\leq \left(\frac{3}{2\log 2-1} + 3\right)  \frac{n_{0}^{p+1}\sqrt{p}}{p!}e^{-\frac{n_{0}}{\xi}} \\
	&=C_{2}\frac{n_{0}^{p+1}\sqrt{p}}{p!}e^{-\frac{n_{0}}{\xi}},
\end{align}
where 
\begin{equation}
	C_{2} = \left(\frac{3}{2\log 2-1} + 3\right)  < 10.8.
\end{equation}
\end{proof}
\begin{lemma}[\cite{Borwein_2007_Uniform}]\label{lem:inc_gamma_bound}
	Let $\Gamma(a,z)$ be the Incomplete Gamma Function defined in the standard way:
	\begin{equation}
	 	\Gamma(a,z) = \int_{z}^{\infty}e^{-x}x^{a-1}dx.
	 \end{equation}
	Let $z \in \mathbb{R} > (a-1)$.
	Then
	\begin{equation}
		\Gamma(a,z) \leq \frac{z^{a}e^{-z}}{z-(a-1)}.
	\end{equation}
\end{lemma}
\begin{proof}
	Make the substitution $s = \frac{x}{z}-1$.
	Then
	\begin{equation}
		\Gamma(a,z) = \int_{0}^{\infty} e^{-(s+1)z}z^{a}(1+s)^{a-1}ds = z^{a}e^{-z}\int_{0}^{\infty}e^{-sz}(1+s)^{a-1}ds.
	\end{equation}
	From here, $(1+s)\leq e^{s}$ implies that
	\begin{equation}
		\Gamma(a,z) \leq z^{a}e^{-z}\int_{0}^{\infty}e^{-sz}e^{(a-1)s}ds = \frac{z^{a}e^{-z}}{-z+a-1} e^{-(z-(a-1))s}\bigg|_{s=0}^{\infty} = \frac{z^{a}e^{-z}}{z-(a-1)},
	\end{equation}
	as long as $z > a-1$ so that the upper limit actually vanishes. 
\end{proof}

\begin{lemma}\label{lem:true_truncated_difference}
	The difference between the truncated and true Hamiltonian obeys 
	\begin{align}
	 \norm{H -\tilde{H}} \leq C_{U}N^{2}e^{-\frac{r_{U}}{2\xi}} + C_{J}Nr_{J}e^{-kr_{J}}.
	\end{align}
\end{lemma}
\begin{proof}
A straightforward application of the triangle inequality yields
\begin{equation}\label{eq:triangle_appendix}
	\norm{H-\tilde{H}} \leq \sum_{I}\abs{(J_{I}-\tilde{J}_{I})} + \abs{\tilde{J}_{I}}\norm{(\tau_{I}^{z} - \tilde{\tau}_{I}^{z})}.
\end{equation}

Recall that the truncated coefficients $\tilde{J}_{I}$ are 0 beyond range $r_{J}$).
In the sum below, the symbol $p$ represents how many sites are coupled by $J$.
That is, the relevant term is $J_{i_1,\ldots i_p}$, a $p$-body term.
The symbol $\ell$ denotes the maximum distance between any two sites coupled by a term of this form, given by $\ell=\abs{i_1-i_p}$.
The first term of \cref{eq:triangle_appendix} may be bounded as follows:
\begin{align}
\sum_{I}\abs{J_{I}-\tilde{J}_{I}} & \leq \sum_{p=2}^{r_{J}}N\sum_{\ell=r_{J}}^{\infty}\binom{\ell-1}{p-2}e^{-\frac{\ell}{\xi}} + \sum_{p=r_{J}+1}^{\infty}N\sum_{\ell=p-1}^{\infty}\binom{\ell-1}{p-2}e^{-\frac{\ell}{\xi}} \\
& \leq N\sum_{p=0}^{r_{J}-2}S_{p,r_{J}-1}{e^{-1/\xi}} + N \sum_{p = r_{J}-1}^{\infty}S_{p,p}{e^{-1/\xi}} \\
& \leq \frac{Ne^{-1/\xi}}{1-e^{-\kappa}} e^{-\kappa (r_{J}-1)} + CNe^{-1/\xi}\sum_{p=r_{J}-1}^{\infty}pe^{-ap} 
\\ & \leq \frac{Ne^{-1/\xi}}{1-e^{-\kappa}} e^{-\kappa (r_{J}-1)} + CNe^{-1/\xi} \left[\frac{(r_J-1)e^{-a(r_J-1)}}{1-e^{-a}} + \frac{e^{-a-a(r_J-1)}}{(1-e^{-a})^2} \right]
\\ & = \frac{Ne^{-1/\xi}}{1-e^{-\kappa}} e^{-\kappa (r_{J}-1)} + CNe^{-1/\xi} \frac{e^{-a r_J}}{(1-e^{-a})^2} \times (e^a(r_J-1)-r_J+2)
\\ &\leq c_1 N e^{-\kappa r_{J}} + c_2 N r_J e^{-ar_J}
\\ &\leq C_{J}N r_J e^{-kr_{J}},
\end{align}
where 
\begin{align}
	c_{1} &= \frac{e^{\kappa}e^{-1/\xi}}{1-e^{-\kappa}} = \frac{1}{2(1-e^{-\kappa})},\\
	c_{2} &= Ce^{-1/\xi}\frac{e^{a}+1}{(1-e^{-a})^{2}} = C \frac{1}{(1-e^{-a})^{2}},\\
	C &= 10.8.
\end{align}
$C_{J}$ is a constant that is independent of $N$ but will depend on $\xi$ (directly and through $a$ and $\kappa$), and
\begin{align}
	\kappa &\equiv \frac{1}{\xi} - \log 2, \\
	a &\equiv \log(e^{1/\xi}-1), \\
	k &\equiv \min\left\lbrace \kappa, a\right\rbrace.
\end{align}
The requirements on both $a$ and $\kappa$ are the same, $\xi < \frac{1}{\log 2}$. 

To bound the second term, we first use a telescoping sum, the triangle inequality, and unitary invariance of the operator norm to show that
\begin{equation}
  	\norm{(\tau_{I}^{z} - \tilde{\tau}_{I}^{z})} \leq \sum_{j = 1}^{p}\norm{(\tau_{i_{j}}^{z} - \tilde{\tau}_{i_{j}}^{z})} \leq 8\sqrt{q}Np e^{-\frac{r_{U}}{2\xi}},
\end{equation}  
where $I$ is the multi-index $i_{1}\dots i_{p}$. 
Plugging this back in yields
\begin{align}
	\sum_{I}\abs{\tilde{J_{I}}}\norm{\tau_{I}^{z}-\tilde{\tau}_{I}^{z}} &\leq \sum_{\ell=1}^{r_{J}-1}N\sum_{p=2}^{\ell+1}\binom{\ell-1}{p-2}8p\sqrt{q}Ne^{-\frac{r_{U}}{2\xi}}e^{-\frac{\ell}{\xi}} \\
	&=8\sqrt{q}e^{-\frac{1}{\xi}}N^{2}e^{-\frac{r_{U}}{2\xi}}\sum_{\ell=0}^{r_{J}-2}\sum_{p=0}^{\ell}\binom{\ell}{p}(p+2)e^{-\frac{\ell}{\xi}} \\
	&\leq 8\sqrt{q}e^{-\frac{1}{\xi}}N^{2}e^{-\frac{r_{U}}{2\xi}}\sum_{p=0}^{r_{J}-2}(p+2)\sum_{\ell=p}^{r_{J}-2}\binom{\ell}{p}e^{-\frac{\ell}{\xi}} \\
	&\leq 8\sqrt{q}e^{-\frac{1}{\xi}}N^{2}e^{-\frac{r_{U}}{2\xi}}\sum_{p=0}^{r_{J}-2}(p+2)S_{p,p} \\
	&\leq 8\sqrt{q}e^{-\frac{1}{\xi}}N^{2}e^{-\frac{r_{U}}{2\xi}}\sum_{p=0}^{r_{J}-2}C(p+2)pe^{-ap} \\
	&\leq C_{U}N^{2}e^{-\frac{r_{U}}{2\xi}}
\end{align}
for some constant $C_{U}$.
In the second-to-last line, we have bounded $S_{p,p}$ using \cref{lem:sum_bound_full}.

Thus, altogether, we have that:
\begin{equation}
	\norm{\Delta H} \leq C_{U}N^{2}e^{-\frac{r_{U}}{2\xi}} + C_{J}Nr_{J}e^{-kr_{J}}.
\end{equation}
\end{proof}
\end{document}